\documentclass[a4paper,10pt]{article}
\usepackage[utf8x]{inputenc}
\usepackage{mathrsfs}
\usepackage{amssymb}
\usepackage{amsmath}
\usepackage{amsthm}
\usepackage{stmaryrd}
\usepackage{enumerate}

\newtheorem{Th}{Theorem}[section]
\newtheorem{prop}[Th]{Proposition}
\newtheorem{lemme}[Th]{Lemma}
\newtheorem{cor}[Th]{Corollary}
\newtheorem{hyp}[Th]{Hypothesis}
\newtheorem{rem}[Th]{Remark}
\newenvironment{dem}{\textbf{Proof}\,:\quad}{\hfill$\square$\bigbreak}
\numberwithin{equation}{section}

\def\aaa{{\mathcal A}}\def\ccc{{\mathcal C}}\def\ddd{{\mathcal D}}
 \def\hhh{{\mathcal H}}

\def\mmm{{\mathcal M}} 
\def\sss{{\mathcal S}}

\def\R{\mathbb R}\def\C{\mathbb C}\def\N{\mathbb N}

\def\D{\partial}\def\eps{\varepsilon}%\def\phi{\varphi}
\def\norm#1{\left\Vert#1\right\Vert}

\def\id{\mathrm{Id}} \def\re{\mathrm{Re}\,} 
\def\Hess{\mathrm{Hess}\,}

\title{Small eigenvalues of the low temperature linear relaxation Boltzmann equation with a confining potential}
\author{Virgile Robbe\thanks{Laboratoire de math\'ematiques Jean Leray, Universit\'e de Nantes, 
2, rue de la Houssini\`ere - BP 92208 F-44322 Nantes Cedex 3. This work is supported by the ANR project NOSEVOL, ANR 2011 BS01019 01}}

\begin{document}

\maketitle

\begin{abstract}
We study the linear relaxation Boltzmann equation, a simple semiclassical kinetic model. We provide 
a resolvent estimate for an associated non-selfadjoint operator as well as an estimate on the return to equilibrium. This is done using a scaling argument and 
non-semiclassical hypocoercive estimate.

\end{abstract}

 

\section{Introduction}

The spectral study of inhomogeneous kinetic equations and the trend to the equilibrium of the related system of particles
are  natural subjects of interest and some progress has been made in the past decade in the spirit of the 
so-called 'hypocoercivity'. In this article we are interested in the study at low temperature of a simple (from the kinetic 
point of view) but difficult (from the spectral analysis point of view) linear model where collisions between particles are not
of diffusion type, but of (non-local) relaxation type. 
 This system as already been studied by Hérau 
in \cite{He06} with improvements by Dolbeault {\it et al.} in \cite{DoMoSc09,DoMoSc10}, but at fixed temperature.\\
The final purpose is to study the existence  of metastable states and a possible 
tunnel effect for the system, which implies very long relaxation time to the equilibrium.
Here we provide first spectral results for the low-lying eigenvalues and the return to the equilibrium at low temperature
 for the following simple linear relaxation Boltzmann model:

$$\left\{
          \begin{array}{l}
           \D_tf+v.\D_xf-\frac1m\D_xV.\D_vf=Q(f)\\
f_{|t=0}=f_0\\
          \end{array}\right.
,$$
where the unknown $f(t,x,v)$ is the density of probability of the system of particles at time $t\in\R_+$, position $x\in\R^d$
and velocity $v\in\R^d$. We will assume that, for all $t\geq0$, $f(t,\,\cdot\,,\,\cdot\,)$ belongs to $L^2$.  Here the collision kernel $Q$ models interactions between particles in the gas and is given
by $$Q(f)=\gamma \left[\left(\int_{\R_v^d} f(t,x,v)\,\mathrm dv\right)\mathrm m_{\beta}-f\right],$$
where $$\mathrm m_{\beta}(v)=\displaystyle{\frac {\mathrm e^{-\frac {m \beta v^2} {2}}}{(\frac{2\pi}{m\beta})^{d/2}}}$$ 
is the $L^1(\R^d_v)$-normalized Maxwellian in the velocity direction with $\beta=1/kT$ where $k$ is the Boltzmann constant, $T$ the temperature of 
the system, and $\gamma$ is the friction coefficient.
 The potential $V\in\ccc^\infty(\R^d,\R)$ only depends on the position $x$.
 We are interested in the low temperature regime of the system and develop for this a semiclassical framework.
We put $h=kT=1/\beta$, we set $m=1$, $\gamma=1$, and pose
$$\mu_h=\mathrm m_{1/h}=\frac1{(2\pi h)^{d/2}}\mathrm e^{-\frac{v^2}{2h}}.$$
We also introduce the spatial Maxwellian and the global Maxwellian given by 
$$\rho_{h}(x)=\displaystyle{\frac {\mathrm e^{-\frac {V(x)}{h}}} {\displaystyle{\int_{\R_x^d}\mathrm e^{-\frac {V(x)}{h}}}}}.  
\qquad  \mmm_h(x,v)=\rho_{h}(x)\mu_{h}(v).$$
It is immediate to check that $\mmm_h$  is the only distributional steady state of the system up to renormalization (see \cite{HelNi05}).
This equation describes a system of large number of particles submitted to an external force coming from the potential $V$
and interacting according to the collision kernel $Q$, whose effect is a simple relaxation toward the local Maxwellian
$\left(\int_{\R_{v'}^d} f(t,x,v')\,\mathrm dv'\right)\mu_{h}(v)$.

We multiply our equation by $h$ which becomes

$$\left\{
          \begin{array}{l}
           h\D_tf+v.h\D_xf-\D_xV.h\D_vf=hQ(f),\\
f_{|t=0}=f_0.\\
          \end{array}\right.
$$
The semiclassical limit $h\rightarrow0^+$ corresponds to the low temperature 
regime of the system.
\begin{hyp}\label{hypo}
 The potential $V$ is a Morse function with $n_0$ local minima
and with derivatives of order 2 or more bounded.
Moreover $\mathrm e^{-\frac V h}\in L^1$ for all $h\in[0,h_0[$ and
there exists $C>0$ such that $|\nabla V(x)|\geq \displaystyle{\frac1C}$ for $|x|>C$. 
\end{hyp}

A standard Hilbert space for the study of the time independent equation is the weighted space
$$\hhh=\left\{f\in\ddd'\;|\;\frac f {\mmm_h^{1/2}}\in L^2\right\}$$
and we see from the Cauchy-Schwarz inequality an Hypothesis 1.1 that $\hhh$ is a subset of $L^1(\R_x^d\times\R_v^d,
\mathrm dx\mathrm dv).$
This is then more convenient to work with the rescaled function $$u(t,x,v)= \mmm_h^{-1/2}f (t,x,v)\in\ccc(\R_+,L^2)$$
(the continuity will be a consequence of an associated semigroup property)
and the new equation reads
\[\left\{
          \begin{array}{l}
           h\D_tu+v.h\D_xu-\D_xV.h\D_vu+ h(\id-\Pi_h)u=0\\
u_{|t=0}=u_0,\\
          \end{array}\right.
\]
where $\Pi_h$ is the orthogonal projection in $\hhh$ (with $t$ as a parameter) on the space 
$E_h=\left\{\rho\mu_{h}^{1/2},\;\rho\in L^2(\R^d_x)\right\}$
of local equilibria (note that $E_h$ is closed). In deed, we have 
$$Q(\mmm_h^{1/2})=\mmm^{1/2}(u-\langle u,\mu_h\rangle_{L^2(\R^d_v)}\mu_h)$$
and
$$\rho_h(x)^{-1/2}v.\partial_x\rho_h^{1/2}(x)-\mu_h(v)^{-1/2}\partial_xV.\partial_v\mu_h^{1/2}(v)=0$$
 The time independent operator is now
        $$\begin{array}{lll}P_h&=&v.h\D_x-\D_xV.h\D_v+ h(\id-\Pi_h)\\
                             &=&X_0^h+ h(\id-\Pi_h),
          \end{array}$$
and the aim of this paper is to prove the following theorem:
\begin{Th}\label{main} Suppose that $V$ satisfies hypothesis \ref{hyp}. Then $P_h$ has $0$ as simple eigenvalue and there 
exists $ h_0>0,$ and $\delta>0$ such that
\begin{enumerate}[i)]
           
           \item for all $h\in]0,h_0]$, $\mathrm{Spec}\,P_h\cap B(0,\delta h)$ consists of exactly $n_0$ (counted for multiplicity) real  eigenvalues which are exponentially small
with respect to $\frac1h$,
           \item for all $\delta_1\in]0,\delta],$ there exists $C>0$ such that, for all $h\in]0,h_0]$, if $\delta_1h\leq |z|\leq\delta h$ then
 $$\|(P_h-z)^{-1}\|\leq \frac C h.$$         
          \end{enumerate}
\end{Th}

This work is a first step towards getting semiclassical expansions of the smallest eigenvalues of our operator 
in the spirit of  Helffer {\it et al.} in \cite{HelKlNi04} for the Witten Laplacian and
by Hérau {\it et al.} in \cite{HeSjSt05,HeHiSj1,HeHiSj2,HeHiSj3} 
for the Kramers-Fokker-Planck operator. In those four last papers, the authors obtained  asymptotic expansions for 
low-lying eigenvalues of the operator
using pseudodifferential technics. Here, one of the major problem is that the projection $\Pi_h$ in $P_h$
does not obey good symbol estimates uniformly in $h$. To overcome this problem, we use a scaling $S_h$ defined below in (\ref{S_h}) and $h$-free hypocoercive
estimates in the spirit of \cite{He06,DoMoSc09,DoMoSc10} for the linear relaxation Boltzmann equation or \cite{HeNi04} for the 
Fokker-Planck equation. Since we work in a non-selfadjoint case, some particular tools have to be used. We first need to get resolvent estimates allowing
control of the norm of the spectral projection onto the generalized eigenspace associated with 
all the eigenvalues with modulus smaller than $\delta h$.
We also need the notion of a PT-like symmetry (cf. \cite{HeHiSj3}) which is a powerful tool that we use like in \cite{HeHiSj3} in
 order to show that there are no 
Jordan blocks in the action of the operator on the generalized eigenspaces associated with the low-lying eigenvalues.

The plan of the paper is the following:
in the second section we prove a hypocoercive estimate from which we deduce a resolvent estimate. In the third  section we finish proving the main result using the PT symmetry property.
The last part is dedicated, in the spirit of Boltzmann H-theorem, to getting a convergence result of 
the solution of the linear relaxation Boltzmann equation on the generalized eigenspace associated with the low-lying
eigenvalues.\bigskip

\textbf{Notation:} For any multi-index $\alpha\in\N^d$, we write
$\partial^\alpha=\partial_{x_1}^{\alpha_1}...\partial_{x_d}^{\alpha_d}$ (and similarly for $v$). We also denote 
$u=\widetilde O(\mathrm e^{-\frac\alpha h})$ for $u\in L^2$ when there exist $C,\,N>0$ such that 
$\|u\|\leq h^{-N}\mathrm e^{-\frac\alpha h}$

\section{Hilbertian hypocoercivity }

Hilbertian hypocoercivity (cf. \cite{DoMoSc09,He06,Vi09}) refers to way to get coercive estimates by using a slight modification
of the scalar product or the operator.
We shall first discuss the maximal accretivity of our operator, so that we can apply our spectral results to describe the
properties of the semigroup.
As $P_h$ is the sum of a non-negative selfadjoint operator (an orthogonal projection) and a skewadjoint operator,
so $P_h$ is accretive.
Let equip $P_h$ with the domain $D=\left\{u\in L^2\;|\;X_0^hu\in L^2\right\}$, we then get a maximal accretive operator.
Indeed $X_0^h$ is maximal accretive on $D$ and $\id-\Pi_h$ is a bounded operator on $L^2$. So $P_h=X_0^h+h(\id-\Pi_h)$
is maximal accretive.

As announced in the introduction, the proof of our estimate relies on a scaling argument. If we conjugate $P_h$
with the dilatation operator, 
\begin{equation}\label{S_h}S_h\,:\left\{\,\begin{array}{ccc}L^2(\R_x^d\times\R_v^d)&\rightarrow&L^2(\R_x^d\times\R_v^d)\\
                              u&\mapsto& h^{-d/2}u(\frac . {\sqrt{h}},\frac . {\sqrt{h}})\end{array}\right. ,\end{equation}
then $S_h^{-1}P_hS_h=hP$ where $P$ is the operator

        $$\begin{array}{lll}P&=&v.\D_x-\D_xV_h(x).\D_v+ (\id-\Pi_1)\\
                             &=&X_0+ (\id-\Pi_1),
          \end{array}$$
where $V_h(x)=\frac 1 h V(\sqrt h x)$.
Notice that $P$ depend on $h$ because the potentials $V_h$ depend on $h$. Nevertheless we chose a $h$-less notation, since 
we will have estimates for $P$ that are uniform w.r.t. $h$.
More precisely, these estimates will only depend on the $L^\infty$ norm of second order or higher order derivatives of 
$V_h$.
Using that derivatives of order two 
or more of $V$ are bounded, we have (if $h\leq1$) for $k\geq2$ and $\alpha\in\N^d$ with $|\alpha|=k$, 
$$\|\partial_x^\alpha V_h\|_\infty=h^{(k-2)/2}\|\partial_x^\alpha V\|_\infty\leq\|\partial_x^\alpha V\|_\infty.$$
Thus we get uniformity w.r.t $h$ on the $L^\infty$ norm of second order or higher order derivatives of 
$V_h$.

 We now will follow \cite{He06} to get a hypocoercive estimate on $P$.
We introduce the two following differential operators: 
$$\begin{array}{cc} a_j=(\D_{x_j}+\D_{x_j}V_h/2),&\quad b_j=(\D_{v_j}+v_j/2)
           \end{array},$$
and their formal adjoints
 $$\begin{array}{cc} a_j^*=(-\D_{x_j}+\D_{x_j}V_h/2),&\quad b_j^*=(-\D_{v_j}+v_j/2)
           \end{array}.$$
We put $$\begin{array}{cc} a=\left(\begin{array}{c} a_1\\ \vdots\\a_d\end{array}\right)
                        &\quad b=\left(\begin{array}{c} b_1\\ \vdots\\b_d\end{array}\right) 
                                       \end{array},$$

and $$\Lambda^2=a^*a+b^*b+1.$$
Notice that $a^*a=-\Delta_x+|\D_xV_h(x)|^2/4-\Delta V_h(x)/2$ is nothing but the Witten Laplacian (in position) 
and  $b^*b=-\Delta_v+v^2/4-d/2$ the harmonic oscillator (in velocity).
Under our hypotheses (see \cite{HelNi05}), $\mathscr{S}(\R^{2d})$ is a core for $\Lambda^2$ and $\Lambda^r$ which
is well-defined for all $r\in\R$ and the operators $a$, $b$, and $\Lambda^2$ are continuous on
$\sss$ and $\sss'$. We have the following relation between $a$, $b$, and $P$:
$$P=b^*a-a^*b+(\id-\Pi_1)$$

We also introduce the semiclassical scaling operator acting on functions only depending on $x$ (to be related to (\ref{S_h})), 
that is $$T_h\,:\left\{\,\begin{array}{ccc}L^2(\R_x^d)&\rightarrow&L^2(\R_x^d)\\
                              u&\mapsto& h^{-d/4}u(\frac . {\sqrt{h}})\end{array}\right. .$$
Notice that $W_h=hT_ha^*aT_h^{-1}$ is then the semiclassical Witten Laplacian: $$W_h=-h^2\Delta_x+|\D_xV(x)|^2/4-h\Delta V(x)/2.$$
Under our hypotheses \ref{hypo} on $V$ and from \cite{HelKlNi04},  this Witten Laplacian $W_h$ has $n_0$ exponentially small
real eigenvalues and there exists $0<\tau\leq1$ fixed 
from now on such that the remaining part of the spectrum is included in $[\tau h,+\infty[$ for $h\in]0,1]$.
We recall from \cite{HelKlNi04} that there exist well-chosen  cut-off functions $\chi_j\in\ccc^\infty_0$, $1\leq j\leq n_0$ each localizing to a neighborhood of $j^\mathrm{th}$ minimum of $V$, for which
\begin{equation}\label{quasW}e_j(x)=\chi_j(x)\mathrm e^{-\frac{V(x)}{2h}}\end{equation}
serves as a quasimode for $W_h$ in the sense that, for some $c>0$,
$$W_he_j=O(\mathrm e^{-\frac ch}\norm{e_j}).$$
We defined the associated non-semiclassical quasimodes $$f_j(x)=T_h^{-1}e_j(x), \qquad g_j(x,v)=f_j(x)\mu_{1}^{1/2}(v),$$ (that still depend on $h$)
and introduce the associated vector spaces
\begin{equation*}\begin{array}{c}
F=\mathrm{Vect}\left\{f_j\,,\;j\in\llbracket1...n_0\rrbracket\right\}\subset L^2(\R_x^d),\\
G=\mathrm{Vect}\left\{g_j\,,\;j\in\llbracket1...n_0\rrbracket\right\}\subset L^2(\R_x^d\times\R_v^d).
   \end{array}\end{equation*}

We also put for all $j$
\begin{equation}\label{quas}
g_j^h=S_hg_j
\end{equation} 

We have the following lemma:

\begin{lemme}\label{Quas}
 The family $(g_j^h)_j$ is almost orthogonal and consists of exponentially small quasimodes for $P_h$ (respectively $P_h^*$), meaning that there exists $\alpha>0$ such that 
for all $j$, $k\in\llbracket1...n_0\rrbracket$, $j\neq k$, $$(g_j^h,g_k^h)=O(\mathrm e^{-\frac \alpha h}\|g_j^h\|\|g_k^h\|)$$.
and for all $j\in\llbracket1...n_0\rrbracket$, $$P_hg_j^h=\widetilde O(\mathrm e^{-\frac \alpha h}\|g_j^h\|)$$ (respectively $P_h^*g_j^h=\widetilde O(\mathrm e^{-\frac \alpha h}\|g_j^h\|)$).
\end{lemme}

\begin{proof} According to the expression of the quasimodes in (\ref{quasW}), we have
 $$g_j^h(x,v)=e_j(x)\mu_{h}^{1/2}(v)=\chi_j(x)\mathrm e^{-\frac{V(x)} {2h} }\mathrm e^{-\frac{v^2}{4h}}.$$
 We immediately deduce the relation of almost orthogonality of the family $(g_j^h)_j$ from the one of $(e_j)_j$ (see  Proposition (6.1) in \cite{HelKlNi04}).
From the expression of $g_j^h$, we deduce that $P_hg_j^h=X_0^hg_j^h=v.\nabla\chi_j\mathrm e^{-\frac{V(x)} {2h} }\mathrm e^{-\frac{v^2}{4h}}$.
From the estimates on $\chi_j$ (see proof of Proposition (6.1) in \cite{HelKlNi04}), we get that $\|X_0^hg_j^h\|=O(\mathrm e^{-\frac \alpha h})$ which shows the assertion for $P_h$
Since $P_h^*=-X_0^h+ h(1-\Pi_h)$, we get the same result for $P_h^*$.
\end{proof}

The following proposition is  the core of the hilbertian hypocoercivity and expresses a coercivity property of the operator $P$
using a small bounded perturbation involving the fundamental auxiliary operator $$L=\Lambda^{-2}a^*b.$$ 

\begin{prop}\label{hyp}There exist $\eps,\,A,\,h_0>0$  such that for all $h\leq h_0$ and
$u\in\sss\cap G^\perp$
                 $$\re(Pu,(\id+\eps(L+L^*))u)\geq\frac{1} {A} \|u\|^2,$$
where $A$ can be chosen to depend explicitly on the second and third derivatives of $V$ and
$\|\eps L\|\leq1$.
\end{prop}

\begin{rem} \textup{ In \cite{He06}, the case of one dimensional $G$ is treated. Here $V_h$ satisfy similar hypotheses, but
the spectral gap is dramatically (exponentially) small due to the fact that the minima of $V_h$ are at a distance of order
$\frac1{\sqrt h}$ from one another. Using $G^\perp$ allows us to get a bound in proposition \ref{hyp} which is uniform with respect to $h$ sufficiently small.}
\end{rem}

\begin{dem} We follow partially \cite{He06} and take care of the uniform 
dependence with respect to $h$.
  Let $u\in L^2$ and $\eps>0$. We have 
               $$\begin{array}{l}\re(Pu,(\id+\eps(L+L^*))u)\\
                                 = \re((\id-\Pi_1)u,(\id+\eps(L+L^*))u)
                                            +\mathrm{Re}(X_0u,(\id+\eps(L+L^*))u)\\    
                                 =\|(\id-\Pi_1)u\|^2+\eps\re((\id-\Pi_1)u,(L+L^*)u)
                                              +\eps\re(X_0u,(L+L^*)u)\\
                                =I+II+III,
                 \end{array}$$
where we used that $X_0$ is skewadjoint for the last term.
We first get a rough lower bound for the two first terms with the Cauchy-Schwarz inequality: 
          $$I+II\geq\frac12\|(\id-\Pi_1)u\|^2-\frac12 \eps^2\|(L+L^*)u\|^2\geq\frac12\|(\id-\Pi_1)u\|^2-\eps^2\|L\|^2\|u\|^2.$$
Let now study the last term more carefully:
     $$III=\eps\re(X_0u,(L+L^*)u)=\eps\re([L,X_0]u,u),$$
since $X_0$ is skewadjoint. Using the definition of $L=\Lambda^{-2}a^*b$ and the commutation relations between $a$, $b$, 
$\Lambda^2$, and $X_0$ (cf. \cite{He06}), we get
               $$\begin{array}{lll}[L,X_0]&=&[\Lambda^{-2}a^*b,X_0]\\
                                          &=&[\Lambda^{-2},X_0]a^*b+\Lambda^{-2 }[a^*,X_0]b+\Lambda^{-2}a^*[b,X_0]\\
                                          &=&-\Lambda^{-2}[\Lambda^{2},X_0]\Lambda^{-2}a^*b-\Lambda^{-2}b^*\Hess V_h b+\Lambda^{-2}a^*a
                  .\end{array}$$
We used the algebraic relation $[A,B^{-1}]=-B^{-1}[A,B]B^{-1}$. 
We put $$\aaa=-\Lambda^{-2}[\Lambda^{2},X_0]\Lambda^{-2}a^*b-\Lambda^{-2}b^*\Hess V b.$$
We are now going to prove to two intermediate lemmas.

\begin{lemme}
The operators $\aaa$ and $L$ are bounded on $L^2$ (uniformly in $h$). 
Moreover their norm can be explicitly bounded in terms of the second and third order derivatives of $V$.
\end{lemme}
\begin{rem} \textup{In fact these bounds are with respect to the  second and third order derivatives of $V_h$, 
but this bounds are uniform in $h$, because the second and third order derivatives of $V_h$ are uniformly bounded
with respect to $h$.}
\end{rem}
\begin{dem} We recall some ideas of the proof from \cite{He06}. Using that $L=\Lambda^{-2}a^*b$, we compute the commutator 
$$[\Lambda^2,X_0]=-b^*(\Hess V_h-\mathrm{Id})a-a^*(\Hess V_h-\id)b,$$ 
which gives
       $$\begin{array}{ll}\aaa=&\Lambda^{-2}b^*(\Hess V_h-\id)a\Lambda^{-2}a^*b+\Lambda^{-2}a^*(\Hess V_h-\id)b\Lambda^{-2}b^*a\\
                               &-\Lambda^{-2}b^*\Hess V_hb.
          
         \end{array}
$$
We see that it is sufficient to prove that for any real $d\times d$ matrix $M(x)$ which depends only on $x$
and which is bounded with its first derivative bounded, the operators
        $$\Lambda^{-2}b^*M(x)a,\quad\Lambda^{-2}b^*M(x)b,\quad\Lambda^{-2}a^*M(x)b$$
are bounded on $L^2$. We only prove boundedness  for the first operator, since the proofs for the other two operators are similar and easier.
It is sufficient to show boundedness for the adjoint $a^*M(x)b\Lambda^{-2}$.
For $u\in\sss$, we write
           $$\begin{array}{lll}\\\|a^*M(x)b\Lambda^{-2}u\|&\leq&\displaystyle{\sum_{j,k}}\|a_j^*M_{j,k}(x)b_k\Lambda^{-2}u\|\\
              &\leq&\displaystyle{\sum_{j,k}}\|M_{j,k}(x)a_j^*b_k\Lambda^{-2}u\|+\displaystyle{\sum_{j,k}}\|(\D_jM_{j,k}(x))b_k\Lambda^{-2}u\|\\
              &\leq&\left(\|M\|_{L^\infty}+\|\nabla M\|_{L^\infty}\right)\\
                   &&\times   \displaystyle{\sum_{j,k}}(\|a_j^*b_k\Lambda^{-2}u\|+\|b_k\Lambda^{-2}u\|),
             \end{array}$$
where we used that $[a_j^*,M_{j,k}]=-\D_jM_{j,k}$. As $b^*b\leq\Lambda^2$ and $1\leq\Lambda^2$, we easily check that
$\|b_k\Lambda^{-2 }u\|\leq\|u\|$. For the term $\|a_j^*b_k\Lambda^{-2}u\|$, we write
          $$\begin{array}{l}
             \|a_j^*b_k\Lambda^{-2}u\|^2=(a_ja_j^*b_k\Lambda^{-2}u,b_k\Lambda^{-2}u)\\
             \leq (a_j^*a_jb_k\Lambda^{-2}u,b_k\Lambda^{-2}u)+((\D^2_jV_h)b_k\Lambda^{-2}u,b_k\Lambda^{-2}u)\\
             \leq (\Lambda^2b_k\Lambda^{-2}u,b_k\Lambda^{-2}u)+\|\Hess V_h\|_{L^\infty}(b_k\Lambda^{-2}u,b_k\Lambda^{-2}u).
            \end{array}
$$
We now use that $[\Lambda^2,b_k]=- b_k$ and continue to derive our series of inequalities:
$$\begin{array}{l}
   (\Lambda^2b_k\Lambda^{-2}u,b_k\Lambda^{-2}u)+\|\Hess V_h\|_{L^\infty}(b_k\Lambda^{-2}u,b_k\Lambda^{-2}u)\\
   \leq (b_ku,b_k\Lambda^{-2}u)+(\|\Hess V_h\|_{L^\infty}+1)(b_k\Lambda^{-2}u,b_k\Lambda^{-2}u)\\
   \leq(\|\Hess V_h\|_{L^\infty}+2)\|u\|^2,
  \end{array}
$$
because $b_k^*b_k\leq\Lambda^2$ and $1\leq\Lambda^2$. And this ends the proof of the lemma thanks to the previous remark.
\end{dem}

Let us go back to the proof of proposition \ref{hyp}. We have $E_1\subset\ker b$ because $b$ is the annihilation operator in velocity. 
Because $b$ appears on the right of every term in the definition of $\aaa$, we therefore get $$E_1\subset\ker\aaa.$$
So we can write  $\aaa=\aaa(\id-\Pi_1)$ and we have for $u\in G^\perp\cap\sss$
                   $$\begin{array}{lll}III&=&\eps\re(\aaa(\id-\Pi_1)u,u)+\eps\re(\Lambda^{-2}a^*au,u)\\
                                         &\geq&-\frac14\|(\id-\Pi_1)u\|^2-\eps^2\|\aaa\|^2\|u\|^2+\eps\re(\Lambda^{-2}a^*au,u).
                     \end{array}$$
It is clear that $\Lambda^2$, $a^*a$, and $\Pi_1$ commute and we can write
\begin{equation}\label{**}
\begin{array}{lll}\eps\re(\Lambda^{-2}a^*au,u)&=&\eps\re(\Lambda^{-2}a^*a\Pi_1u,u)+\eps\re(\Lambda^{-2}a^*a(\id-\Pi_1)u,u)\\
             &=&\eps\re(\Lambda^{-2}a^*a\Pi_1u,\Pi_1u)\\
             &&+\eps\re(\Lambda^{-2}a^*a(\id-\Pi_1)u,(\id-\Pi_1)u)\\
             &\geq&\eps\re(\Lambda^{-2}a^*a\Pi_1u,\Pi_1u)-\eps\|(\id-\Pi_1)u\|^2,
  \end{array}\end{equation}
where we used that $a^*a\leq\Lambda^2$.
Here we slightly diverge from \cite{He06}: we are going to work with the quasimodes of the Witten Laplacian instead of
the first eigenvector of $\Lambda^2$ and this will give us exponentially small remainder terms in our estimates.
We now need an intermediate lemma.
\begin{lemme}\label{gap}
For all $u\in G^\perp\cap\sss$ then, 
$\re(\Lambda^{-2}a^*a\Pi_1u,\Pi_1u)\geq\frac\tau 4\|\Pi_1u\|^2$.
\end{lemme}

\begin{dem}
Let $\mathbb P_h$ denote the spectral projection onto the eigenspaces associated with the $n_0$ smallest
eigenvalues of the semiclassical Witten laplacian $W_h=hT_ha^*aT_h^{-1}$(which are the same as the eigenvalues of $ha^*a$).
According to \cite{HelKlNi04} and \cite{HelSj84}, we have for all 
 $w\in\sss(\R_x^d)$
$$(hT_ha^*aT_h^{-1}(1-\mathbb P_h)T_hw,(1-\mathbb P_h)T_hw)\geq \tau h\|(1-\mathbb P_h)T_hw\|^2$$

We then put $\mathbb P=T_h^{-1}\mathbb P_hT_h$  the (orthogonal) projection on the spectral subspace associated with the $n_0$ smallest
eigenvalues of the Witten laplacian $a^*a$, so the previous inequality becomes (since $T_h$ is unitary)
$$(a^*a(1-\mathbb P)w,(1-\mathbb P)w)\geq \tau \|(1-\mathbb P)w\|^2$$
Moreover, we have
   $$\begin{array}{lll}(a^*aw,w)&=&(a^*a[1-\mathbb P +\mathbb P]w,[1-\mathbb P +\mathbb P]w)\\
      &=&(a^*a(1-\mathbb P)w,(1-\mathbb P)w)+(a^*a\mathbb Pw,\mathbb Pw)\\
&&+(a^*a(1-\mathbb P)w,\mathbb Pw)+(a^*a\mathbb Pw,(1-\mathbb P)w)\\
&=&(a^*a(1-\mathbb P)w,(1-\mathbb P)w)+(a^*a\mathbb Pw,\mathbb Pw),
     \end{array}
$$
because the ranges of $1-\mathbb P$ and $\mathbb P$ are stable under $a^*a$.
By definition of $\mathbb P$, we get that there exists $\alpha>0$ such that  for all $j$
$$(a^*a\mathbb Pw,\mathbb Pw)=O(\mathrm e^{-\frac\alpha h})\|w\|^2.$$
Now, if we take $w\in F^\perp$, we have $\|(1-\mathbb P)w\|^2=\|w\|^2+O(\mathrm e^{-\frac\alpha h})\|w\|^2$.\\
We therefore get $$(a^*aw,w)\geq\tau\|w\|^2+O(\mathrm e^{-\frac\alpha h})\|w\|^2.$$
Since $(a^*a+1)^{-1/2}w\in F^\perp$, we get by the max-min principle (since $a^*a$ is selfadjoint)
             $$(a^*a(a^*a+1)^{-1/2}w,(a^*a+1)^{-1/2}w)\geq\frac\tau{1+\tau}\|w\|^2+O(\mathrm e^{-\frac\alpha h})\|w\|^2.$$
Moreover we clearly have $$1\geq\frac \tau{1+\tau}\geq\frac \tau2,$$
because $\tau\leq1$. We get the result by taking
as $w$ the function defined for almost every $v$ by $x\mapsto\Pi_1u(x,v)$ and $h$ small enough, 
since $\Pi_1u\in E_1$ and $\Pi_1u(\,\cdot\,,v)\in F^\perp$.
\end{dem}
\textbf{End of proof of proposition \ref{hyp}:}
We can now put the result of Lemma \ref{gap} in inequality (\ref{**}) and we get
               $$\eps\re(\Lambda^{-2}a^*au,u)\geq\eps\frac\tau{4}\|\Pi_1u\|^2-\eps\|(\id-\Pi_1)u\|^2.$$
We then get
$$III\geq-\frac 14 \|(\id-\Pi_1)u\|^2-\eps^2\|\aaa\|^2\|u\|^2+\eps\frac\tau{4}\|\Pi_1u\|^2-\eps\|(\id-\Pi_1)u\|^2.$$
Bringing together our estimates for $I+II$ and $III$, this yields to
  $$\begin{array}{l}\re(Pu,(\id+\eps(L+L^*))u)\\
     \geq\frac 12 \|(\id-\Pi_1)u\|^2-\eps^2\|L\|^2\|u\|^2\\
     -\frac 14 \|(\id-\Pi_1)u\|^2-\eps^2\|\aaa\|^2\|u\|^2+\eps\frac\tau{4}\|\Pi_1u\|^2-\eps\|(\id-\Pi_1)u\|^2\\
     \geq \frac 18 \|(\id-\Pi_1)u\|^2+\eps\frac\tau{4}\|\Pi_1u\|^2-\eps^2(\|\aaa\|^2+\|L\|^2)\|u\|^2
   . \end{array}$$
by taking $\eps\leq1/8$. Using that $\Pi_1$ is an orthogonal projection and that $\eps\leq1/8$ we get
           $$\re(Pu,(\id+\eps(L+L^*))u)\geq\left(\eps\frac\tau{4}-\eps^2(\|\aaa\|^2+\|L\|^2)\right)\|u\|^2.$$
Taking $\eps/\tau$ small enough, we get for a constant $A$ sufficiently large (uniform in $h$)
$$\re(Pu,(\id+\eps(L+L^*))u)\geq\frac{\tau^2}A\|u\|^2.$$
This ends the proof of the proposition.
\end{dem}
\begin{cor}\label{Res}
 There exists $c>0$ and $h_0>0$ sufficiently small, such that for all $h\in[0,h_0[$, for all $z\in\C$ with $\re z\leq ch$ and $v\in\sss\cap (S_hG)^\perp$
\begin{equation}\label{res} \|(P_h-z)v\|\geq ch\|v\|.\end{equation}
\end{cor}
\begin{dem}
We just proved that for all $u\in\sss\cap G^\perp$
 $$\re(Pu,(\id+\eps(L+L^*))u)\geq\frac{\tau^2}A\|u\|^2.$$
We then get by multiplying the inequality by $h$ and putting $v=S_hu$
$$\re(P_hv,S_h(\id+\eps(L+L^*))u)\geq\frac{\tau^2h}A\|u\|^2.$$
For all $z\in\C$, we can write
$$\re((P_h-z)v,S_h(\id+\eps(L+L^*))u)\geq\frac{\tau^2h}{A}\|u\|^2-\re(z(v,S_h(\id+\eps(L+L^*))u)).$$
First notice that if we choose $\eps<\displaystyle{\frac1{2\|L\|}}$, we then have $(v,S_h(\id+\eps(L+L^*))u)\geq0$.
So by the Cauchy-Schwarz inequality, we get
\begin{eqnarray*}
  \lefteqn{\norm{(P_h-z)v}\norm{S_h(\id+\eps(L+L^*))u}}\\&&\geq\displaystyle{\frac{\tau^2h}{A}}\|u\|^2
-\min(0,\re z)\norm{v}\norm{S_h(\id+\eps(L+L^*))u}.
  \end{eqnarray*}
Since $S_h$ is unitary and $L$ is bounded, we get for all $ v\in\sss\cap(S_hG)^\perp$
\begin{eqnarray*}
  \lefteqn{\norm{(P_h-z)v}(1+2\eps\|L\|)\norm{v}}\\
&&\geq\frac{\tau^2h}{A}\|v\|^2-\min(0,\re z)(1+2\eps\|L\|)\norm{v}^2.
\end{eqnarray*}
Therefore, taking $\re z\leq ch= \displaystyle{\frac{\tau^2h}{2(1+2\eps\|L\|)A}}$ leads to
$$ \|(P_h-z)v\|\geq ch\|v\|.$$

\end{dem}

Since we are working with a non-selfadjoint operator, we have to be careful in order to get a resolvent estimate on the whole space.
 Thanks to lemma \ref{Quas}, we see that $S_hG=\mathrm{Vect}\,g_j^h$ vanishes under $P_h$ up to a $O(\mathrm e^{-\frac \alpha h})$ and that $(S_hG)^\perp$ is stable under $P_h$ up to a 
$O(\mathrm e^{-\frac \alpha h})$, meaning that there exists $\alpha>0$ such that for all $u\in S_hG$ and $v\in(S_hG)^\perp$,
$$P_hu=\widetilde O(\mathrm e^{-\frac{\alpha} h}\|u\|),$$
and there exists $v'\in L^2$ with $v'=\widetilde O(\mathrm e^{-\frac{\alpha} h}\|v\|)$ such that
$$P_hv-v'\in(S_hG)^\perp.$$
 Denoting by $\Pi$ the orthogonal projection on $S_hG$, the two previous relations and Pythagoras' theorem  allow us to write that for all $u\in L^2$
 $$\begin{array}{lll}
 \|(P_h-z)(\id-\Pi)u+(P_h-z)\Pi u\|^2&=&\|(P_h-z)(\id-\Pi)u\|^2\\
 &&+\|(P_h-z)\Pi u\|^2+O(\mathrm e^{-\frac{2\alpha} h})\|u\|^2.
 \end{array}$$
We then have for all $u\in L^2$ 
   $$\begin{array}{lll}\|(P_h-z)u\|^2&=&\|(P_h-z)(\id-\Pi)u+(P_h-z)\Pi u\|^2\\
      &=& \|(P_h-z)(\id-\Pi)u\|^2+\|(P_h-z)\Pi u\|^2\\&&+O(\mathrm e^{-\frac{2\alpha} h})\|u\|^2\\
     &\geq& \frac {c^2 h^2}2\|(1-\Pi)u\|^2+\|(P_h-z)\Pi u\|^2+O(\mathrm e^{-\frac{2\alpha} h})\|u\|^2
     ,\end{array}$$
where we used the estimate (\ref{res}) and choose $z$ smaller than $\frac c {\sqrt2}$.
Thanks to the triangle inequality and the fact that $P_h\Pi u=O(\mathrm e^{-\frac\alpha h})\|u\|$ by definition of $\Pi$, we get
 \begin{eqnarray*}\begin{array}{lll}\|(P_h-z)u\|^2&\geq&\frac {c^2 h^2}2\|(1-\Pi)u\|^2+|z|^2\|\Pi u\|^2+O(\mathrm e^{-\frac{2\alpha} h})\|u\|^2\\
 &\geq&\min(\frac {c^2 h^2}2,|z|^2)(\|(1-\Pi)u\|^2+|z|^2\|\Pi u\|^2)
     ,\end{array}
     \end{eqnarray*}
Therefore, for all $\delta_1<\frac c2$, if we take $\delta_1 h\leq|z|\leq \frac c2 h$ and $h$ small enough we get
$$\|(P_h-z)u\|\geq \frac{\delta_1} 2\|u\| $$
We conclude that there exists $\delta>0$ such that for all $0<\delta_1\leq\delta$, if 
$\delta_1 h\leq|z|\leq\delta h$, we have the resolvent estimate $$\|(P_h-z)^{-1}\|\leq \frac {C_{\delta_1}}h.$$
This completes the proof of part \textit{ii)} of Theorem \ref{main}.
\section{PT-symmetry}

Let $H$ be the generalised eigenspace associated with the eigenvalues of modulus smaller than $\delta h$, where $\delta$ is sufficiently small as in Theorem \ref{main}, 
that is the range of the spectral projection $\Pi_0=\frac 1 {2i\pi}\int_{|z|=\delta h}(z-P_h)^{-1}\,\mathrm dz$.
We can write the proposition
\begin{prop} 
We have $ \mathrm{dim}\, H=n_0$, where $n_0$ is the number of local minima of the potential $V$.
\end{prop}
\begin {dem}
From proposition \ref{hyp}, the coercive estimate on the orthogonal of a $n_0$ dimensional subspace
gives us that $\mathrm{dim}\, H\leq n_0$. Moreover we saw that $(z-P_h)g_j^h=zg_j^h+\widetilde O(\mathrm e^{-\frac\alpha h})$; for $z$ with $|z |=\delta h$,  we therefore get $(z-P_h)^{-1}g_j^h=\frac1zg_j^h+\widetilde O(\mathrm e^{-\frac\alpha h})$. 
By integrating along the circle centered in 0 and with radius $\delta h$, we finally get
 $$\Pi_0 g_j^h=g_j^h+\widetilde O(\mathrm e^{-\frac\alpha {2h}})$$
Since the $g_j^h$ are almost orthogonal, they are linearly independent. We deduce that $\mathrm{dim}\,H\geq n_0$, so $\mathrm{dim}\,H=n_0$.\end{dem}
It remains to show that $H$ contains no Jordan blocks for $P_h$ which we do by using an extra symmetry of our operator that is PT-symmetry: P referring to parity (that is velocity reversal) and T to time reversal.
Let $\kappa\,:\,\R^{2d}\rightarrow\R^{2d}$ be given by $\kappa(x,v)=(x,-v)$. We then put $U_\kappa u=u\circ\kappa,\,u\in L^2$, so that 
$U_\kappa$ is unitary and selfadjoint. We also introduce the non-definite Hermitian form $$(u,v)_\kappa=(U_\kappa u,v),\;u,\,v\in L^2$$
Note that $$P_h^*=U_\kappa^{-1}P_hU_\kappa,$$
so $P_h$ is formally selfadjoint with respect to the Hermitian form $(.,.)_\kappa$.
\begin{prop}For $h$ small enough, the restriction of $(.,.)_\kappa$ to $H\times H$ is positive definite uniformly with respect
 to $h$.
\end{prop}
\begin{dem} Let us recall that $g_j^h(x,v)=\chi_j(x)\mmm^{1/2}_h(x,v)$ and that $\Pi_0 g_j^h=g_j^h+\widetilde O(\mathrm e^{-\frac\alpha {2h}})$ . 
So $g_j^h\circ\kappa=g_j^h$ and the family $(g_j^h)_{1\leq j\leq n_0}$ is almost orthonormal for the Hermitian form $(.,.)_\kappa$ 
Then there exists a basis $(g_{0,j})_{1\leq j\leq n_0}$ of $H$  such that
$$(g_{0,j},g_{0,j})_\kappa=1+O(\mathrm e^{-\frac \alpha h}),\quad (g_{0,j},g_{0,j'})_\kappa=O(\mathrm e^{-\frac \alpha h})\,\mathrm{pour}\,j\neq j'.$$
Then for $H\ni u=\displaystyle{\sum_{k=1}^{n_0}u_kg_{0,k}}$, we get
              $$(u,u)_\kappa=\sum_{k=1}^{n_0}(1+O(\mathrm e^{-\frac \alpha h}))|u_k|^2\geq\|u\|^2/C.$$
This ends the proof. 
\end{dem}
Thanks to the formal selfadjointness of $P_h$ with respect to  $(.,.)_\kappa$, we get the following proposition.
\begin{prop} For $h$ small enough, the restriction of $P_h\,:\,H\rightarrow H$ is selfadjoint with respect to the scalar product $(.,.)_\kappa$ on $H$. Moreover $P_h$ has exactly $n_0$ real eigenvalues counted for multiplicity smaller than $\delta h $ and they are $O(\mathrm e^{-\frac\alpha h})$.
\end{prop}
This implies in particular point \textit{i)} of Theorem \ref{main}.
\section{Return to equilibrium}
To conclude, we translate our spectral estimates into estimates for the semigroup decay.
Let us recall that Corollary \ref{Res} gives us that for all $u\in\sss\cap(S_hG)^\perp$ and $\re z\leq ch$
         $$\|(P_h-z)u\|\geq c h  \|u\|.$$

Let us rewrite the previous resolvent estimate as an estimate on $(\id-\Pi_0)L^2$ for $\Pi_0$ the spectral projection associated with the $n_0$ smallest eigenvalues. We first note that $(g_j^h)_j$ are quasimodes 
for both $P_h$ and $P_h^*$. We then have for all $j$
$$\Pi_0g_j^h=\Pi_0^*g_j^h=g_j^h+\widetilde O(\mathrm e^{-\frac \alpha h}).$$
So if we take $u\in L^2$, we get for all $j$
$$(g_j^h,(\id-\Pi_0)u)=((\id-\Pi_0^*)g_j^h,(\id-\Pi_0)u)=O(\mathrm e^{-\frac \alpha h})\|(\id-\Pi_0)u\|$$
We conclude that one can write for all $u\in L^2$
$$(\id-\Pi_0)u=v+\displaystyle{\sum_j a_j g_j^h}$$
with $v\in G^\perp$, $\|v\|=\|(\id-\Pi_0)u\|+O(\mathrm e^{-\frac \alpha h})\|(\id-\Pi_0)u\|$ and for all $j$, 
$|a_j|=O(\mathrm e^{-\frac \alpha h})\|(\id-\Pi_0)u\|$.
We then have with $\re z\leq ch$, thanks to our previous resolvent estimate,
$$\|(P_h-z)(\id-\Pi_0)u\|\geq ch\|v\|-O(\mathrm e^{-\frac \alpha h})\|(\id-\Pi_0)u\|$$
We therefore get for $h$ small enough and a new $c$ the resolvent estimate on $(\id-\Pi_0)L^2$ with $\re z\leq ch$
          $$\|(P_{h|(\id-\Pi_0)L^2}-z)^{-1}\|\leq\frac 1{ch}.$$
Therefore we have obtained an uniform resolvent estimate on the left half-plane $\left\{\re z\leq ch\right\}$.
As $\Pi_0L^2$ and $(\id-\Pi_0)L^2$ are stable under $P_h$ (by definition of $\Pi_0$), the Gearhart-Prüss theorem 
(cf. \cite{Hel13}) for $P_{h\,|(\id-\Pi_0)L^2}$ allows us to write that for all $t>0$
 $$\mathrm e^{-tP_h}=\mathrm e^{-tP_h}\Pi_0+O(e^{-ch t}).$$ Despite the fact that the projection is not orthogonal,
since we have $\Pi_0=\frac 1 {2i\pi}\int_{\ccc(0,\delta h)}(z-P_h)^{-1}\,\mathrm dz$, our resolvent estimate gives us 
$\|\Pi_0\|=O(1)$. Now if we denote by $(\mu_j)_{j=1...n_0}$ the exponentially small eigenvalues (counted with multiplicity)
and $\Pi_j$ the associated spectral projections, we have that $\|\Pi_j\|=O(1)$.
Indeed, we saw that the restriction $P_{h\,|H}$ is selfadjoint with respect to $(.,.)_\kappa$,
which is equivalent (on $H$) to the ambient scalar product, so the projections are bounded.
Thus, we can sum up all of this into the following Proposition.
\begin{prop} There exists $\delta_1>0$ and $\alpha>0$ such that for all , $t\geq0$ and $h$ small enough
               $$\mathrm e^{-tP_h}=\sum_{j=1}^{n_0}\mathrm e^{-t\mu_j}\Pi_j+ O(e^{-\delta_1h t}),$$
with $\|\Pi_j\|=O(1)$ and $\mu_j=O(\mathrm e^{-\frac\alpha h})$.

\end{prop}

\noindent\textbf{Acknowledgments:} We want to deeply thank J. Viola for his useful remarks concerning both mathematical and grammatical questions in the first version of this paper.

\nocite{Vi02}
\bibliography{biblio}
\bibliographystyle{siam}
\end{document}